\documentclass[aps,pra,twocolumn,reprint]{revtex4-1}
\usepackage{graphicx}
\usepackage{dcolumn}
\usepackage{bm}
\usepackage{epsfig} 	
\usepackage{epsf}
\usepackage{amsmath}
\usepackage{amssymb}
\usepackage{amsfonts}
\usepackage{amscd}
\usepackage{amsthm}

\usepackage{wasysym}
\usepackage{listings}
\usepackage{cases}
\usepackage{url}
\usepackage{times}

\usepackage[utf8]{inputenc}
\usepackage{color}

\definecolor{gris}{gray}{0.60}
\definecolor{rouge}{rgb}{1,0,0}
\definecolor{jaune}{rgb}{1,1,0}
\newcommand{\be}{\begin{equation}}
\newcommand{\ee}{\end{equation}}
\newcommand{\bfg}{\begin{figure}}
\newcommand{\efg}{\end{figure}}
\newcommand{\bra}[1]{\langle#1|}
\newcommand{\ket}[1]{|#1\rangle}

\newcommand{\ketbra}[2]{|#1\rangle\langle#2|}

\newcommand{\C}{\mathbb{C}}
\newcommand{\Z}{\mathbb{Z}}

\newcommand{\gra}{{\boldsymbol \alpha}}
\newcommand{\grb}{{\boldsymbol \beta}}
\newcommand{\grg}{{\boldsymbol \gamma}}
\newcommand{\grz}{{\boldsymbol 0}}
\newcommand{\gri}{{\boldsymbol i}}
\newcommand{\grj}{{\boldsymbol j}}
\newcommand{\grk}{{\boldsymbol k}}

\newcommand{\ie}{{\em i.e. }}
\newcommand{\floor}[1]{\lfloor #1 \rfloor}
\newcommand{\vl}{\vec{\lambda}}

\newtheorem{myprop}{Proposition}
\newenvironment{mypropbis}[1]
  {\renewcommand{\themyprop}{\ref{#1}${\bf bis}$}%
   \addtocounter{myprop}{-1}%
   \begin{myprop}}
  {\end{myprop}}
\newtheorem{mydef}{Definition}

\newtheorem{mytheorem}{Theorem}

\begin{document}
\sloppy
\title{Exploring pure quantum states with maximally mixed reductions
}
\author{Ludovic Arnaud}
\email{Electronic address: larnaud@ulb.ac.be}
\author{Nicolas J. Cerf}
\affiliation{QuIC, Ecole Polytechnique de Bruxelles, CP 165/59, Universit\'e Libre de Bruxelles, 1050 Brussels, Belgium.}
\begin{abstract}
We investigate multipartite entanglement for composite quantum systems in a pure state. Using the generalized Bloch representation for $n$-qubit states, we express the condition that all $k$-qubit reductions of the whole system are maximally mixed, reflecting maximum bipartite entanglement across all $k$~vs.~$n-k$ bipartitions. As a special case, we examine the class of balanced pure states, which are constructed from a subset of the Pauli group  $\mathcal{P}_n$ that is isomorphic to $\Z_2^n$. This makes a connection with the theory of quantum error-correcting codes and provides bounds on the largest allowed $k$ for fixed $n$. In particular, the ratio $k/n$ can be lower and upper bounded in the asymptotic regime, implying that there must exist multipartite entangled states with at least $k=\floor{0.189 \, n}$ when $n\to \infty$.  We also analyze symmetric states as another natural class of states with high multipartite entanglement and prove that, surprisingly, they cannot have all maximally mixed $k$-qubit reductions with $k>1$. Thus, measured through bipartite entanglement across all bipartitions, symmetric states cannot exhibit large entanglement. However, we show that the permutation symmetry only constrains some components of the generalized Bloch vector, so that very specific patterns in this vector may be allowed even though $k>1$ is forbidden. This is illustrated numerically for a few symmetric states that maximize geometric entanglement, revealing some interesting structures.

\end{abstract}
\maketitle


\section{Introduction}
Quantum entanglement is certainly one of the most fascinating concepts arising in quantum mechanics, essentially because it appears as a contradiction to reductionism, {\it i.e.}, the principle by which understanding a complex system reduces to the description of each of its individual constituents \cite{Schroedinger1935}. As a matter of fact, a quantum composite system can possibly be in a state such that its parts are more disordered -- have a higher entropy -- than the whole system. This peculiar property, known as the non-monotonicity of the von Neumann entropy \cite{Wehrl1978}, is tightly linked to the notion of {\em bipartite entanglement}. A pure bipartite entangled state, for example, admits a zero entropy, which translates the fact that one has complete knowledge about the joint system via its wavefunction. Its two parts, however, are mixed, so that they exhibit a nonzero entropy. In other words, one knows less about the parts than about the system taken as a whole, a property which cannot be conceived in classical terms.

The essence of bipartite entanglement is thus that the information about a quantum bipartite system is not only encoded in its parts, but also in the correlations between them. Remarkably, when a bipartite quantum system is {\em maximally entangled}, the information appears to be fully encoded in these correlations and no longer in the system's constituents. Mathematically speaking, while the whole system is described as a pure state, its parts are individually described as maximally mixed states (with a density matrix proportional to the identity). A paradigmatic example of such a situation is the Einstein-Poldolsky-Rosen (EPR) state of two qubits \cite{EPR,Bohm}, which is a pure bipartite state whose parts are maximally mixed: each qubit has an entropy of 1 bit, so its state is completely unknown, while the 2-qubit joint state is perfectly determined. Equivalently, one observes that the entropy of one part conditionally on the other is negative (it is $- 1$ bit), which is a sufficient condition for bipartite entanglement and can be associated with a flow of (virtual) information backwards in time \cite{Cerf1997}.

A very intriguing question is of course whether similar situations may exist if the system is made out of more than two parts. The underlying concept of {\em multipartite entanglement} becomes naturally much richer than bipartite entanglement, but also generally much more difficult to understand (see, e.g., \cite{RMP-Horodecki2009} for a review on entanglement). It leads to stronger contradictions with local realism than bipartite entanglement\cite{Greenberger1990}, as well as to the existence of several inequivalent classes of entangled states even in the simplest case of three qubits \cite{Dur2000}. Multipartite entanglement is also crucial to applications, such as one-way quantum computing \cite{Raussendorf2001}, and its dynamics when exposed to a dissipative environment has revealed a surprisingly large variety of flavors \cite{Aolita2008,Barreiro2010}. 

Among the possible approaches to multipartite entanglement, one of them consists of probing the presence of bipartite entanglement over all inequivalent {\it bipartitions} of all sizes \cite{Facchi2008}. Roughly speaking, the idea is to measure how much each subset of $k$ out of $n$ constituents (with $0<k \le \lfloor n/2 \rfloor$) can be bipartite entangled with its $n-k$ complementary constituents, knowing that there is a subtle balance with the bipartite entanglement exhibited by all other possible subsets with respect to their complements. This leads to the concept of a {\em genuine} multipartite entangled state, that is, a $n$-partite pure state such that none of its $k$-partite subsets can be represented by a pure state (all subsets are mixed, hence bipartite entangled with their complements). 

One may even be more specific and seek for a {\em strong form} of a genuine multipartite entangled state. This would be a composite system in a pure state such that all of its individual constituents are {\em maximally} mixed, all pairs of its constituents are {\em maximally} mixed, all triplets of its constituents are {\em maximally} mixed, and so on up to all $k$-tuples of its constituents. This property, namely the fact of admitting {\em maximally mixed reductions}, is an ideal case of genuine multipartite entanglement. We expect that the constraint of having an overall pure state will set a limit on the possibility of having maximally mixed $k$-tuples for large $k$. Consider, for example, a system of $n$ qubits in an overall pure state. It is tempting to search for $n$-qubit pure states having the property that all subsets of $k$ qubits are maximally mixed up to size $k= \lfloor n/2 \rfloor$. It has long been known, however, that except for a very few low-dimensional cases, such as the EPR state for $n=2$, it is impossible to find a $n$-qubit pure state exhibiting this strong form of genuine multipartite entanglement \cite{Scott2004}.

In this paper, we examine this fundamental question and investigate multipartite entanglement under two perspectives. First, we focus on the possible existence of $n$-qubit states that satisfy this property of admitting maximally-mixed $k$-partite reductions for arbitrary $n$ and $k$. We give examples of known states exhibiting this property, then provide a set of conditions that such states must satisfy as well as asymptotic existence bounds. The second part is centered on a weaker version of this maximally-mixed reduction property, which is motivated by recent results on {\em symmetric states}, a special class of states that are known to exhibit a high multipartite entanglement as measured in terms of their geometric entanglement \cite{Martin2010,Aulbach2010,Markham2011}.

In Section II, we expose the representation of $n$-qubit pure states in terms of a generalized Bloch vector, which is very convenient in order to express the conditions that all $k$-partite reductions are maximally mixed. This leads us to consider, in Sec. III, a class of $n$-partite pure states that we name ``balanced". In the Bloch representation, they correspond to a Bloch vector with all components equal to a same value for indices belonging to some subset of the Pauli group (the other components being all taken equal to zero). This makes a connection with the theory of quantum error correction, from which we obtain lower and upper bounds on the highest allowed value of $k$ for a given value of $n$. In particular, we show that there exist $n$-qubit states admitting all maximally-mixed $k$-partite reductions with at least $k=\floor{0.189 \, n}$ when $n\to \infty$. 

In Section IV, we then consider another natural class of states, namely symmetric states, among which it is known that some states with genuine multipartite entanglement can be found. We prove that, surprisingly, the symmetric states cannot have maximally-mixed reductions of size $k$ that exceed 1, regardless of $n$. In that sense, they are very far from the {\em strong} form of genuine multipartite entanglement that we seek. On the other hand, we show that the permutation symmetry underlying symmetric states only puts constraints on the components of the Bloch vector with an even index, so that the components with an odd index may possibly be taken equal to zero. This brings us to investigate $n$-partite symmetric pure states whose Bloch vector has many vanishing odd-index components. This investigation is carried out numerically, focusing on some symmetric states that maximize geometric entanglement as found in \cite{Aulbach2010}. We show that some of these states are close to having maximally-mixed $k$-partite reductions for large values of $k$ (although, strictly speaking, $k=1$), so that they approach the {\em strong form} of genuine multipartite entanglement. Finally, some conclusions are drawn in Sec. V.

\section{Maximally-mixed reduction property}

The property of admitting maximally-mixed reductions is encapsulated by the following definition:
\begin{mydef}
An $n$-qubit pure state $\ket{\psi}$ is a $k$-MM state if all its reductions of size $k$ are maximally mixed. Here and in what follows, MM stands for maximally mixed.
\end{mydef}
Note that according to the definition, a $k$-MM state is also a $(k-\ell)$-MM state for $0\le \ell \le k$, and, in particular, every pure state is a $0$-MM state. A natural question which arises here is whether a $k$-MM state exists for a given couple $(n,k)$. Note that no more than half of the qubits can be in a maximal mixed state, as a consequence of the Hilbert-Schmidt decomposition of the overall pure state. So, it is clear that $k$-MM states cannot exist when $k>\floor{n/2}$. Here, we list the known facts about the existence of $k$-MM states for small values of $n$:

\subsection{$k$-MM states of small size $n$}
\begin{itemize}
\item For $n=2$, it is easy to check that the four Bell states
\begin{equation}
\ket{\Phi^{\pm}}=\frac{\ket{00} \pm \ket{11}}{\sqrt{2}}, \qquad \ket{\Psi^{\pm}}=\frac{\ket{01} \pm \ket{10}}{\sqrt{2}},
\end{equation}
are 1-MM states.
\item For $n=3$, the Greenberger-Horne-Zeilinger (GHZ) state is a 1-MM state, while W state \cite{Dur2000} is not. In general, for any size $n$, the generalized GHZ state
\begin{equation}
\ket{GHZ}=\frac{\ket{00\cdots 0}+\ket{11\cdots 1}}{\sqrt{2}},
\end{equation}
is a 1-MM state.
\item For $n=4$, there exists no 2-MM state. To see this, consider the following 4-qubit states:
\begin{eqnarray}\label{m4}
\ket{L}&=&\frac{1}{2\sqrt{3}}\left( (1-\omega)(\ket{0011} + \ket{1100}) +\omega^2 (\ket{0101} + \ket{0110} \right.\nonumber\\
&&\left. + \ket{1001} + \ket{1010}-\ket{0000}-\ket{1111})\right),\\
\ket{HS}&=&\frac{1}{\sqrt{6}}\left( \ket{0011} + \ket{1100} +\omega(\ket{0101} + \ket{1010})\right.\nonumber\\
&&\left. + \omega^2(\ket{0110} + \ket{1001})\right),\quad \omega=e^{2\pi i/3},
\end{eqnarray}
All of their 1-qubit reductions are maximally mixed, so they are 1-MM state, but this is not true for their 2-qubit reductions. $\ket{HS}$ was introduced in \cite{Higuchi2000} and conjectured to be the 4-qubit maximally entangled state. Then, it was shown to be a local maximum of the averaged 2-qubit von Neumann entropy in \cite{Brierley2007}. Finally, in \cite{Gour2010}, it was shown that the global maximum for all averaged 2-qubit Tsallis and R\'enyi entropies is reached by $\ket{HS}$ for $\alpha<2$, and by the states $\ket{L}$ for $\alpha>2$. 
This implies that no 4-qubit state can have all of its 2-qubit reductions maximally mixed, so indeed no 2-MM state of 4 qubits exists.
\item For $n=5$, the two logical states $\ket{0}_L$ and $\ket{1}_L$ of the 5-qubit code introduced in \cite{laflamme1996} are both 2-MM state. It is easy to check that every qubits (or every pair of qubits) is found in a maximally-mixed state after tracing over the remaining qubits.
\item For $n=6$, the four 6-qubit states constructed as logical Bell state using the previous 5-qubit code states,
\begin{eqnarray}\label{m6}
\ket{M_6^{\phi^{\pm}}}=\frac{\ket{0}\ket{0}_L \pm \ket{1}\ket{1}_L}{\sqrt{2}}, \nonumber \\
\ket{M_6^{\psi^{\pm}}}=\frac{\ket{0}\ket{1}_L \pm \ket{1}\ket{0}_L}{\sqrt{2}}, 
\end{eqnarray}
are 3-MM states.
\item For $n>7$, it is shown in \cite{Scott2004} that no state can have all of its $\floor{n/2}$ reduction maximally mixed. Note that the case $n=7$ is not solved yet.
\end{itemize}
We see that the 2-qubit 1-MM state and 6-qubit 3-MM state appear to be very special cases, with maximally mixed reductions up to precisely half the number of qubits (the trivial  upper bound on $k$ set by the Hilbert-Schmidt decomposition). In general, for an arbitrary number $n$ of qubits, we can expect that it will be possible to find a $k$-MM state up to some threshold value $k_{\mathrm{max}}<\floor{n/2}$, which only depends on $n$.

\subsection{Generalized Bloch vector formalism}

To analyze this question, we need first to introduce the generalized Pauli matrices, which are the set of matrices constructed in terms of all $n$-fold tensor products of the form
\begin{equation}\label{gpb}
 \sigma_\gra = \sigma_{\alpha_1}\otimes\sigma_{\alpha_2}\otimes...\otimes\sigma_{\alpha_n},
\end{equation}
where each $\sigma_{\alpha_i}$ represents respectively the $2\times 2$ identity matrix or one of the usual Pauli matrices, depending on the index $\alpha_i\in\{0,1,2,3\}$. The bold index $\gra$ refers to a vector index containing the $n$ indices $\alpha_i$. There exist $4^n$ such matrices, all being traceless except for $\sigma_\grz$ which corresponds to the $2^n\times 2^n$ identity matrix and admits a trace $\mathrm{Tr}(\sigma_\grz)=2^n$. Using this set of generalized Pauli matrices $\{\sigma_\gra\}$, it is also possible to construct the bigger set of the form $\{\sigma_\gra,-\sigma_\gra,i\sigma_\gra,-i\sigma_\gra\}$. This set of $4^{n+1}$ elements becomes closed under matrix multiplication and forms the so-called Pauli group $\mathcal{P}_n$.


The set $\{\sigma_\gra\}$ also forms a basis of a complex Schmidt-Hilbert space of dimension $4^n$, so that every complex  square $2^n\times2^n$ matrix can be seen as a vector ${\bf r}$ in this space.  For instance, a matrix $\rho$ reads
\begin{equation}\label{rho}
\rho=\sum_{\gra}r_\gra\,\sigma_\gra \equiv {\bf r},
\end{equation}
while the components $r_\gra$ are given by the inverse formula
\begin{equation}\label{ralpha}
r_\gra=\frac{1}{2^{n}}\mathrm{Tr}(\sigma_\gra\,\rho).
\end{equation}
which, for a pure state, becomes simply
\begin{equation}\label{ralphapure}
r_\gra=\frac{1}{2^{n}}\bra{\psi}\sigma_\gra\ket{\psi}.
\end{equation}
If $\rho$ is a quantum state, then hermiticity ($\rho=\rho^{\dagger}$), positivity ($\rho\ge0$), and normalization ($\mathrm{Tr}\, \rho=1$) give the three following constraints on the components $r_\gra$ \cite{Bengtsson2006}:

\begin{equation}
\begin{cases}
r_\gra \in \mathbb{R},\forall\, \gra,\nonumber\\
{\bf r} \textrm{ is in the positive cone} \nonumber\\ 
r_\grz = \frac{1}{2^{n}}.\nonumber
\end{cases}
\end{equation}
These three relations mean that, after translation by $-1/2^n$ in the zeroth direction, a quantum state is completely represented by the vector ${\bf r}$, which lives in the positive cone contained in a real subspace of a Schmidt-Hilbert space, the so-called {\em generalized Bloch vector} of dimension $4^n-1$. Note that the concept of a positive cone embraces the idea that every convex combination of positive operators is also a positive operator.

Pure states, {\it i.e.}, rank-one projectors that satisfy $\rho^2=\rho$, appear in this representation as vectors ${\bf r}$ that are constrained by 
\begin{eqnarray}
\sum_{\gra\grb}g_{\gra\grb\grg}\,r_\gra r_\grb=r_\grg,\label{pure}
\end{eqnarray}
where $g_{\gra\grb\grg}$ are the structure constants of $SU(2^n)$ defined as $\sigma_\gra\sigma_\grb:=\sum_{\grg}g_{\gra\grb\grg}\,\sigma_{\grg}$. Note that Eq.~\eqref{pure} also automatically implies positivity. It can be decomposed in two independent relations 
\begin{numcases}{}
|\vec{r}|^2\equiv\sum_\gri r_\gri^2=\frac{2^n-1}{2^{n+1}}=R^2,\label{norm}\\
(\vec{r}\star\vec{r})_\gri\equiv\sum_{\grj\grk}g_{\gri(\grj\grk)}r_\grj r_\grk=\frac{2^n-2}{2^n}r_\gri,\label{orientation}
\end{numcases}
where latin indices correspond to vector indices excluding the zeroth component ($\gra\equiv(\grz,\gri)$), $\star$ is by definition the generalization of the cross product, and parentheses stand for symmetrization, $g_{\gri(\grj\grk)}=(g_{\gri\grj\grk}+g_{\gri\grk\grj})/2$. Relations \eqref{norm} and \eqref{orientation} express respectively that for a state being pure, its Bloch vector ${\bf r}$ should live on a sphere of radius $R$ (which is actually the boundary of the positive cone) and should have a specific orientation. For more details on the generalized Bloch representation, see \cite{Bengtsson2006}.

\subsection{Conditions for maximally mixed reductions}

The generalized Bloch representation is a very useful tool in order to address the maximally-mixed reduction property.
Let $\rho$ be the density matrix of an $n$-qubit pure state $\ket{\psi}$ living in the tensor product space $\mathcal{H}=\C^2\otimes\C^2\otimes...\otimes\C^2$, and consider the bipartition $\mathcal{H}=\mathcal{H}_A\otimes\mathcal{H}_B$ where $A$ and $B$ are defined as the sets of the first $k$ qubits and last $n-k$ qubits, respectively. In the generalized Bloch representation, the $k$-qubit reduced density matrix $\rho_{A}$ resulting from tracing out the qubits of B is given by
\begin{eqnarray}
\rho_{A} &=& \mathrm{Tr}_B(\rho) = \mathrm{Tr}_B(\sum_{\gra}r_\gra\,\sigma_\gra)\nonumber\\
&=& \sum_{\gra_{A}\gra_{B}}r_{\gra_{A}\gra_{B}}\,\sigma_{\gra_A}\,\mathrm{Tr}_B(\sigma_{\gra_B}),
\end{eqnarray}
where the index $\gra$ has been decomposed according to the bipartion $\gra_A\gra_B$. By using the fact that the matrices $\sigma_{\gra_B}$ are all traceless excepted for the one corresponding to the identity on $\mathcal{H}_B$, noted $\sigma_{\grz_B}$, we get
\begin{eqnarray}
\rho_A &=& \sum_{\gra_{A}}r_{\gra_{A}\grz_{B}}\,\sigma_{\gra_A}\,\mathrm{Tr}_B(\sigma_{\grz_B})\nonumber\\
&=& 2^{n-k}\sum_{\gra_{A}}\,r_{\gra_{A}\grz_{B}}\,\sigma_{\gra_A},\label{rhoA}
\end{eqnarray}
where $r_{\gra_{A}\grz_{B}}\equiv r_{(\alpha_1\alpha_2\cdots\alpha_k00\cdots0)}$ with $n-k$ zeros in the vector index at the positions corresponding to $B$. If we now consider an arbitrary bipartition $(A,B)$, we obtain an expression similar to \eqref{rhoA} where in the vector index of $r_{\gra_{A}\grz_{B}}$, the zeros are located at the positions of the traced out qubits. For instance, for a 6-qubit state, if the first, third, and last qubits are traced out, the component $r_{\gra_{A}\grz_{B}}$ corresponds to $r_{(0 \alpha_2 0 \alpha_4 \alpha_5 0)}$.

Equation \eqref{rhoA} gives a very nice operational procedure for performing the partial trace. Usually, when expressing the density matrix in the computational basis, the entropy resulting from tracing out a part of the system manifests itself both as a loss of some components and as the mixing of some other components of the original density matrix. In the generalized Bloch representation, the partial trace appears just as the loss of some components (see Fig.~\ref{trace}).
\begin{figure}
\begin{center}
\epsfig{file=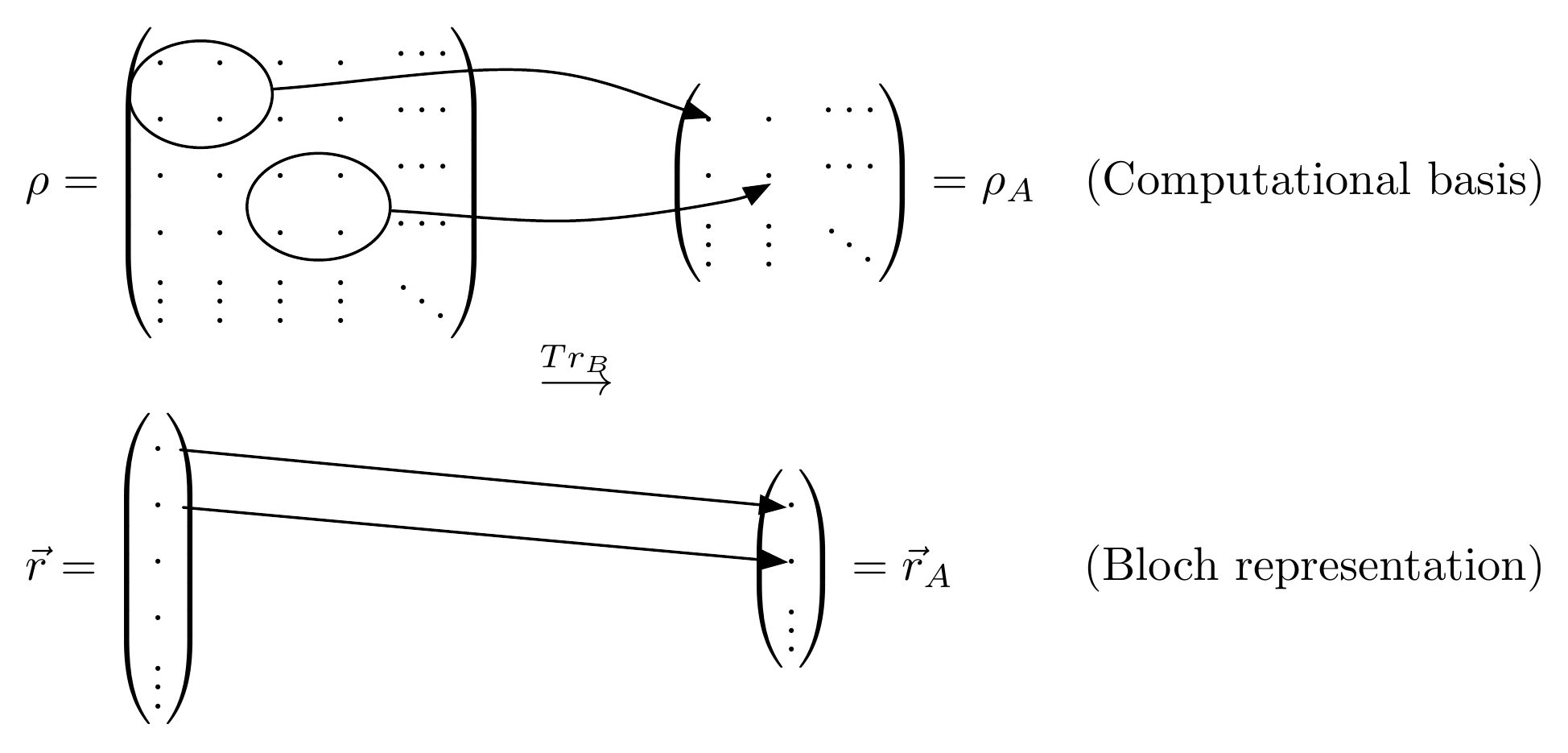,width=8.8cm,angle=0}
\caption{The partial trace seen as a loss of information, both in the computational basis and in the generalized Bloch representation.}\label{trace}  
\end{center}       
\end{figure}
For instance, with the Bloch representation, the linear entropy of the reduced state simply reads
\begin{equation}
S_l(\rho_A)=1-\mathrm{Tr}_A(\rho_A^2)=1-\left(2^{2n-k}\sum_{\gra_A}r_{\gra_A\grz_B}^2\right).
\end{equation} 

Our concern now is the special case where all reduced density matrices $\rho_A$ for a given size $|A|=k$ are proportional to the identity on $\mathcal{H}_A$, {\it i.e.}, when the original state is a $k$-MM state. In view of Eq.~\eqref{rhoA}, this is the case if all components $r_\alpha$ with a vector index $\gra$ containing at least $n-k$ zeros vanish, except for the $r_\grz$ component which is always equal to $1/2^n$. By enumerating all these indices, we see that there are 
\begin{eqnarray}
D_k = \sum_{l=1}^k \binom{n}{l} 3^l \label{dim}
\end{eqnarray}
such components that must vanish. Let us define the {\em weight} $\omega(\gra)$ of an index $\gra$ (or by extension of a component $r_\gra$ or of a generalized Pauli matrix $\sigma_\gra$) by its number of non-zero subindices. Then, we can establish the following proposition:
\begin{myprop}\label{propkmm}
The $n$-qubit state $\rho$ is a $k$-MM state if and only if its corresponding Bloch vectors ${\bf r}$ does not have any component $r_{\gra}$ with an index weight in the range $0 < \omega(\gra) \le k$.
\end{myprop}
This means that the vector ${\bf r}$ (after translation by $-1/2^n$ in the zeroth direction) does not have any non-zero component in the subspace spanned by the basis vectors $\sigma_\gra$ with index weight lower than or equal to $k$. If denote this subspace $\mathcal{E}_k$ (of dimension $D_k$) and $\bar{\mathcal{E}_{k}}$ its orthogonal-complementary subspace, we can establish the equivalent proposition:
\begin{mypropbis}{propkmm}
The $n$-qubit state $\rho$ is a $k$-MM state if and only if its corresponding Bloch vectors ${\bf r}$ has zero components in $\mathcal{E}_k$, so that its support belongs to $\bar{\mathcal{E}_{k}}$.
\end{mypropbis}

The existence of a $k$-MM state results from the compatibility between having a pure state satisfying Eq.~\eqref{pure} and Prop. \ref{propkmm} at the same time. Such a compatibility is not easy to study for an arbitrary couple ($n$,$k$) without more information about the state. For this reason, we focus in the next Section on a class of balanced states that is suitable for analyzing this question and obtaining existence bounds.

\section{Existence bounds for $k$-MM states}

In the Bloch representation, constructing a pure state directly in terms of the components of  its Bloch vector involves the orientation relation \eqref{orientation}, which is difficult to manipulate in general. Even checking the purity of a given state numerically implies $\mathcal{O}(4^{3n})$ operations, and beyond 10 qubits the computation time becomes unreasonable on a standard desktop computer. Instead, we will focus on a restricted class of states, which we call {\em balanced states}.

\subsection{Balanced pure states as $k$-MM states}

Balanced state are defined in the Bloch representation as states whose non-zero components $r_\gra$ have all the same value, namely the same value as the identity component $r_\grz$.
\begin{mydef}
A $n$-qubit balanced state is expressed as
\begin{equation}
\rho_S=\frac{1}{2^n}\sum_{\sigma \in S}\sigma,\label{ring}
\end{equation}
where $S$ is a subset of the Pauli group $\mathcal{P}_n$, which entirely defines the state.
\end{mydef}
Note that in this definition, the hermiticity of $\rho_S$ implies that $S$ does not contain {\em complex} elements of the Pauli group of the form $\pm i\,\sigma$. We can now express the following theorem about the purity of such balanced states:
\begin{mytheorem}\label{theo-balanced}
The $n$-qubit balanced state $\rho_S$ defined from the set $S$ is pure if and only if $S$ forms a group under matrix multiplication that is isomorphic to  $\Z_2^n$.
\end{mytheorem}
\begin{proof}
If $\rho_S$ is a pure state ($\rho_S=\rho_S^2$) then we have
\begin{equation}
\frac{1}{2^n}\sum_{\sigma \in S}\sigma = \frac{1}{2^{2n}}\sum_{\sigma , \tau \in S}\sigma\tau.\label{purering}\\
\end{equation}
The uniqueness of the expansion \eqref{ring} implies that \eqref{purering} is satisfied only when $S$ is closed under matrix multiplication, {\it i.e.}, when $S$ is a subgroup of $\mathcal{P}_n$. It follows that
\begin{eqnarray}
\frac{1}{2^n}\sum_{\sigma \in S}\sigma &=& \frac{1}{2^{2n}}\sum_{\tau\in S}\sum_{\sigma\in S}\sigma \nonumber\\
 &=& \frac{|S|}{2^{2n}}\sum_{\sigma\in S}\sigma,\label{purering2}
\end{eqnarray}
which implies that for $\rho_S$ being pure, the order of $S$ should be $2^n$. In summary, $S$ must fulfill the three following properties:
\begin{itemize}
\item $S$ has a finite order equal to $2^n$;
\item $S$ is abelian because normalization and hermiticity of $\rho_S$ imply $-\sigma_\grz\notin S$ and $\pm i\,\sigma\notin S$, respectively;
\item All the elements of $S$ have order 2, {\it i.e.}, they are such that $\sigma^2=\sigma_0$, since $\pm i\,\sigma\notin S$.
\end{itemize}
According to the fundamental theorem of finite abelian groups \cite{BookGroup}, this implies that $S\simeq\Z_2\times\Z_2\times...\times\Z_2=\Z_2^n$.
\end{proof}
A simple example of such a pure balanced state is the state $\ket{0}\equiv \ket{00\cdots 0}$. Its components in the Bloch representation are given by Eq.~\eqref{ralpha} as 
\begin{equation}
r_\gra=\frac{1}{2^{n}}\mathrm{Tr}(\sigma_\gra\,\ketbra{0}{0})=\frac{1}{2^{n}}\bra{0}\sigma_\gra\ket{0}=\frac{1}{2^{n}}(\sigma_\gra)_{00},\nonumber\\
\end{equation}
which are non-zero if and only if $\sigma_\gra$ is a tensor product of identity and $\sigma_Z$ matrices. There are $2^n$ such matrices $\sigma_\gra$, which form the set $S$, and thus the state can be written explicitly as
\begin{eqnarray}
\ketbra{0}{0}&=&\frac{1}{2^n}\left(\sigma_{(00\cdots00)}+\sigma_{(00\cdots03)}+\sigma_{(00\cdots30)}+\cdots\right.\nonumber\\
 && \left.\cdots + \sigma_{(33\cdots30)}+\sigma_{(33\cdots33)}\right),\label{zero}
\end{eqnarray}
where it is clear that $S\simeq\Z_2^n$ just by relabelling the identity matrix $\sigma_0$ as ``0'' and the Pauli matrix $\sigma_Z=\sigma_3$ as ``1'',
for instance $\sigma_{(00\cdots03)}\rightarrow(00\cdots01)$.

According to Prop.~1, a pure balanced state will be a $k$-MM state if its group $S$ does not contain elements with an index weight lower than or equal to $k$ (except for weight zero). This problem is fully equivalent to finding an additive self-orthogonal quantum error-correcting code over GF(4) \cite{Calderbank1998}. More generally, the connection between entanglement and quantum error correcting codes (QECC) was noted by several authors, and it was proven for example in \cite{Scott2004} that a QECC that can detect $k$ errors is also a $k$-MM state. The reciprocal of this statement can easily be understood by interpreting each Pauli matrix as an error operation. Indeed, for every tensor product of Pauli matrices with index weight $0 < \omega(\gra) \le k$, Eq.~\eqref{ralphapure} gives
\begin{equation}
\bra{\psi}\sigma_\gra\ket{\psi}=0,
\end{equation}
for a $k$-MM state $\ket{\psi}$, which means that each error $\sigma_\gra$ is detectable because it rotates the state in an orthogonal subspace. 

\subsection{Quantum Gilbert-Varshamov and quantum Hamming bounds on $k$-MM states}

By exploiting this relationship, known bounds in the context of QECC can be mapped onto existence bounds for $k$-MM states.
The quantum Gilbert-Varshamov (GV) and quantum Hamming (H) bounds introduced in \cite{Ekert1996} give, respectively, lower and upper bounds on the number of errors that a quantum code can detect, for a given number of qubits. They can be directly translated in our context and give 
\begin{equation}
\underset{(H)}{D_{\floor{k/2}} +1} \le 2^n \le \underset{(GV)}{D_{k} +1}.
\end{equation}
in terms of the subspace dimension $D_k$.
Asymptotically, these bounds can be written as \cite{Ekert1996}
\begin{equation}
\underset{(GV)}{f\left(\frac{k}{n}\right)} \le 0 \le \underset{(H)}{f\left(\frac{k}{2n}\right)},
\end{equation}
where 
\begin{equation}
f(x) = 1-x\,\log_2 3+x\,\log_2x+(1-x)\,\log_2(1-x),
\end{equation}
is a decreasing function which has a root in $x_0\simeq0.18929$. Physically, this means that it is always possible to find a $n$-qubit $k$-MM state ($n,k\to \infty$) that is such that, by keeping less than 19\% or more than 81\% of its qubits, we completely lose the information on the initial pure state. On the other hand, it is impossible to find such a state if we keep between 38\% and 62\% of its qubits. The situation in the region between 19\% and 38\% (or between 62\% and 81\%) is unknown, see Fig.~2. An intriguing physical implication of these bounds is that the entropy behaves as an extensive quantity (it is proportional to the number of qubits) in any subsystem as long as it has a size lower than 19\% of the total system. It is only beyond this bound that, at some point, we observe a defect of extensitivity which originates from the purity of the state of the total system.

Note that there exist more accurate bounds on QECC (see for instance \cite{Feng2004}), but these are useless in our case. Indeed, finding a $k$-MM state is equivalent to finding a one-codeword's code, and typically these other bounds become stronger only when the number of codewords exceeds 1. In our case, there are no known better bounds than the quantum Hamming and quantum Gilbert-Varshamov bounds.

\subsection{Numerical bounds on $k$-MM states}

Constructive upper and lower bounds can also be obtained numerically. For instance, the search of additive self-orthogonal quantum codes over GF(4) based on linear programing has been performed up to around 100 qubits. Such results are plotted in Fig.~2, based on Markus Grassl's database \cite{Grassl:codetables}.
\begin{figure}[h]
\begin{center}
\includegraphics[width=8cm,angle=0]{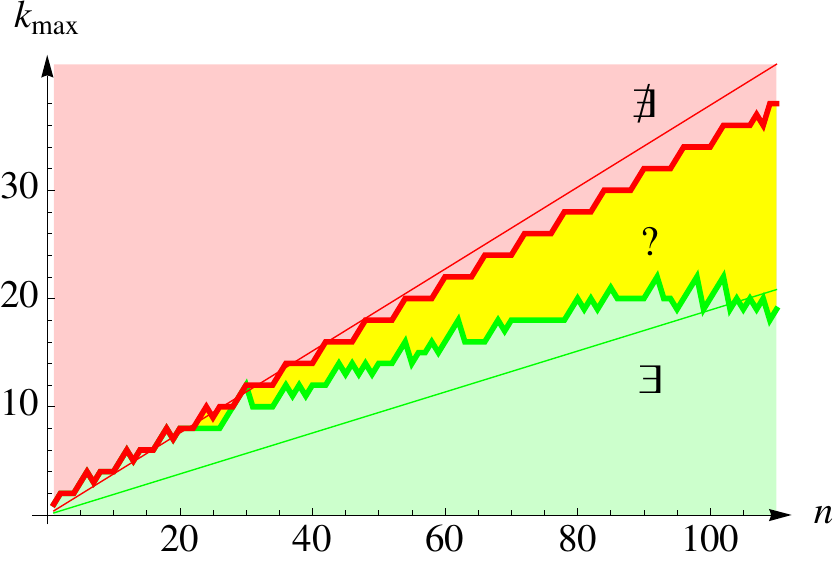}
\caption{(Color online) Constructive upper (thick red line) and lower (thick green line) bounds for the existence of $k$-MM states based on Markus Grassl's database \cite{Grassl:codetables}. The asymptotic limits on the domain of existence are also shown, namely the quantum Hamming bound (thin red line) and quantum Gilbert-Varshamov bound (thin green line).}\label{bounds}  
\end{center}       
\end{figure}

\section{Symmetric $k$-MM states and beyond}

\subsection{Symmetric states}
In \cite{Brown2005}, it is conjectured that any $n$-qubit state of maximal multipartite entanglement should be a $1$-MM state, when the entanglement is measured through (the sum of) the negativity over all inequivalent bipartitions. We may naively extend this conjecture by saying that any $n$-qubit state of maximal multipartite entanglement should be a $k$-MM state, with $k$ being the maximum allowed value as analyzed in the previous Section. While this is very well possible, we will see that the {\em symmetric states}, as defined below, are not the good candidates to test this conjecture. To understand why, we first recall that a symmetric state is a $n$-qubit pure state that is invariant under any permutations of its qubits, that is
\begin{equation}\label{sym}
U_{\pi}\ket{\psi}=\ket{\psi},\forall\, \pi\in S_n,
\end{equation}
where $U_{\pi}$ is the unitary transformation that effects the permutation $\pi$ over the set of qubits and $S_n$ is the symmetric group of $n$ objects. These symmetric states form a $(n+1)$-dimensional symmetric subspace of $\mathcal{H}$, which is often considered as a good subspace to look for genuine multipartite entangled states, particularly in terms of their geometric entanglement \cite{Martin2010,Aulbach2010,Markham2011}. Surprisingly, we observe that they satisfy the following theorem:
\begin{mytheorem}\label{theo-symmetric}
Symmetric $n$-qubit states are at most $1$-MM states.
\end{mytheorem}
\begin{proof}
We first express the unitary transformation that effects the transposition between qubits $i$ and $j$, {\it i.e.}, the {\em SWAP} operator. In terms of generalized Pauli matrices, this operator reads 
\begin{eqnarray}
U_{ij}&=&\frac{1}{2}\left(\sigma_{\grz}+\sigma_{11}^{ij}+\sigma_{22}^{ij}+\sigma_{33}^{ij}\right)\nonumber\\
&=&\frac{1}{2}\left(\sigma_{\grz}+\sum_{a=1}^3\sigma_{aa}^{ij}\right).\label{swap}
\end{eqnarray}
In this expression we use the more usual notation for generalized Pauli matrices, \ie the positions of non-identity matrices are indicated as superscripts while the positions of identity matrices are implicit, and the global identity is still written $\sigma_\grz$. For instance, $\sigma_{11}^{12}$ corresponds to the action of $\sigma_x$ on the first and second qubits. Now, we can write its expectation value in a given pure state $\ket{\psi}$ as 
\begin{eqnarray}
\langle U_{ij}\rangle&=&\frac{1}{2}\left(\bra{\psi}\sigma_{\grz}\ket{\psi}+\sum_{a=1}^3\bra{\psi}\sigma_{aa}^{ij}\ket{\psi}\right)\nonumber\\
&=&\frac{1}{2}\left(1+2^n\sum_{a=1}^3 r_{aa}^{ij}\right).\label{tij}
\end{eqnarray}
Because the symmetric group $S_n$ can be generated by the set of all transpositions of two elements, the set of relations \eqref{sym} is fully equivalent to the set of relations
\begin{eqnarray}\label{casm1}
U_{ij}\ket{\psi}&=&\ket{\psi}\Longleftrightarrow \sum_{a=1}^3 r_{aa}^{ij}=\frac{1}{2^n}, ~ \forall \,i>j \in [1,n],\label{cont1}
\end{eqnarray}
which provide conditions on some components of the generalized Bloch vector ${\bf r}$ of a symmetric state $\ket{\psi}$.
At the same time, a $k$-MM state with $k>1$ should have all components $r_{ab}^{ij}=0$ with $i\ne j$, according to Prop.~1, which necessarily contradicts some of the relations \eqref{cont1}.
\end{proof}
Note that in the context of QECC, Theorem~\ref{theo-symmetric} means that it is not possible to find a code in the symmetric subspace detecting more than one error (or correcting any error). More physically, we see that the permutation symmetry creates some frustration, which prevents the subsets of 2 qubits (or more) to be all maximally mixed. A natural question is of course whether permutation symmetry also manifests itself by constraining {\em higher} index weight terms to be non-zero. To answer this question, let us define the concept of {\em parity} of an index $\gra$ (and by extension of a matrix $\sigma_{\gra}$ or of a component $r_{\gra}$):
\begin{mydef}\label{even-odd}
Let index $\gra$ contain $\lambda_1$ subindices 1, $\lambda_2$ subindices 2, and $\lambda_3$ subindices 3. The parity of $\gra$ is defined as even if the $\lambda_i$'s are all even. Otherwise,  it is defined as odd.
\end{mydef}
For instance the indices $(011)$ and $(1122)$ are even, but $(122)$, $(0123)$ and $(1123)$ are odd. Note that any component $r_\gra$ with an odd index weight $\omega(\gra)$ is necessarily an odd component. We can now state that
\begin{mytheorem}\label{theo-pair}
For a symmetric state, the even components $r_{\gra}$ of a given index weight $w(\gra)$ cannot all vanish.
\end{mytheorem}
\begin{proof}
Let us consider the expectation value of $m$ transpositions acting on disjoined supports ($m\le\floor{n/2}$):
\begin{eqnarray}\label{casm}
\langle \prod_{k=1}^m U_{i_kj_k}\rangle&=&\frac{1}{2^m}\bra{\psi}\prod_{k=1}^m\left(\sigma_{\grz}+\sum_{a_k=1}^3\sigma_{a_ka_k}^{i_kj_k}\right)\ket{\psi}\nonumber\\
 &=&\frac{1}{2^m}\left(1+\sum_{k=1}^m\sum_{a_k=1}^3\bra{\psi}\sigma_{a_ka_k}^{i_kj_k}\ket{\psi}\right.\nonumber\\
&+&\left.\sum_{k<l}\sum_{a_k,a_l}\bra{\psi}\sigma_{a_ka_ka_la_l}^{i_kj_ki_lj_l}\ket{\psi}+\cdots \right),
\end{eqnarray}
where, in this expression, there are $\binom{m}{t}$ terms of the form
\begin{equation}
\sum_{a_1a_2\cdots a_t}\bra{\psi}\sigma_{a_1a_1a_2a_2\cdots a_ta_t}^{i_1j_1i_2j_2\cdots i_tj_t}\ket{\psi},
\end{equation}
which correspond to all the combinations of $t$ transpositions chosen among the $m$ transpositions. By induction, starting from the case $t=1$ corresponding to Eq.~\eqref{casm1}, we conclude that all terms labeled by even indices should be equal to one, leading to the constraints
\begin{equation}\label{general-even}
\sum_{a_1a_2\cdots a_t}r_{a_1a_1a_2a_2\cdots a_ta_t}^{i_1j_1i_2j_2\cdots i_tj_t}=\frac{1}{2^n}.
\end{equation}
for $t=1,2,\ldots m$. Thus for symmetric states, {\em even} components $r_{\gra}$ of weight $2t$ cannot all vanish.
\end{proof}

Note that {\em a priori} the odd components are not constrained explicitly by the permutation symmetry. To see this, a reasoning similar to the above proof can be done by considering the expectation values of a 3-cycle acting on qubits $i_1$, $i_2$ and $i_3$. We have
\begin{eqnarray}\label{cascycle}
\langle U_{i_1i_2}U_{i_2i_3}\rangle&=&\frac
{1}{4}\bra{\psi}(\sigma_{\grz}+\sum_{a=1}^3\sigma_{aa}^{i_1i_2})(\sigma_{\grz}+\sum_{b=1}^3\sigma_{bb}^{i_2i_3})\ket{\psi}\nonumber\\
&=&\frac{1}{4}\left(1+\sum_{a=1}^3\bra{\psi}\sigma_{aa}^{i_1i_2}\ket{\psi}+\sum_{b=1}^3\bra{\psi}\sigma_{bb}^{i_2i_3}\ket{\psi}\right.\nonumber\\
&+&\left.\sum_{ab}\bra{\psi}\sigma_{aa}^{i_1i_2}\sigma_{bb}^{i_2i_3}\ket{\psi}\right)\nonumber\\
&=&\frac{1}{4}\left(1+\sum_{a=1}^3\bra{\psi}\sigma_{aa}^{i_1i_2}\ket{\psi}+\sum_{a=1}^3\bra{\psi}\sigma_{aa}^{i_2i_3}\ket{\psi}\right.\nonumber\\
&+&\sum_{a=b}\bra{\psi}\sigma_{aa}^{i_1i_3}\ket{\psi}+\left.\sum_{a\ne b,c}\epsilon_{abc}\bra{\psi}\sigma_{acb}^{i_1i_2i_3}\ket{\psi}\right).\nonumber\\
\end{eqnarray}
where $\epsilon_{abc}$ stands for the completely antisymmetric symbol.
Since we must have $\langle U_{i_1i_2}U_{i_2i_3}\rangle=1$ and since Eq.~\eqref{casm1} is satisfied, it follows
\begin{eqnarray}
\frac{1}{4}\left(4+\sum_{a\ne b,c}\epsilon_{abc}\bra{\psi}\sigma_{acb}^{i_1i_2i_3}\ket{\psi}\right)&=&1.\nonumber\\
\Longrightarrow\sum_{abc}\epsilon_{abc}r_{acb}^{i_1i_2i_3}=0.\label{odd}
\end{eqnarray}
Since $r_{acb}^{i_1i_2i_3}$ is completely symmetric in its lower indices for a symmetric state, Eq.~(\ref{odd}) is necessarily satisfied. Then, we can generalize the last procedure by averaging the product of disjoined transpositions and a 3-cycle (acting on a disjoined support). Proceeding by induction in analogy to the reasoning leading to Eq. \eqref{general-even}, we obtain again 
a set of relations 
\begin{equation}\label{general-odd}
\sum_{\substack{a_1a_2\cdots a_t\\b_1b_2b_3}}\epsilon_{b_1b_2b_3}r_{a_1a_1a_2a_2\cdots a_ta_t b_1 b_2 b_3}^{i_1j_1i_2j_2\cdots i_tj_t i_1 i_2 i_3}=0
\end{equation}
which are necessarily satisfied. Thus, in contrast with the {\em even} components which cannot all vanish in a symmetric state, the {\em odd} components are not constrained by permutation symmetry.

It must be stressed that the odd components may, however, be possibly constrained for another reason, in particular as a result of the purity constraint \eqref{pure}. It is nevertheless natural to seek for symmetric states such that a large number of their odd components would vanish, which would correspond to high multipartite entanglement. In order to test this possibility, we now investigate specific states defined in \cite{Martin2010,Aulbach2010,Markham2011}, known to exhibit high geometric entanglement.
%

\subsection{Symmetric states with high geometric entanglement}
In \cite{Martin2010,Aulbach2010,Markham2011}, some symmetric states with high geometric entanglement have been found. In particular, in \cite{Aulbach2010}, an optimization procedure was performed up to $n=12$ in order to find the states that maximize their geometric entanglement.  Those states are good candidate to test whether odd component $r_\gra$ indeed vanish. Of course, a direct calculation of these components by using expression \eqref{ralphapure} in the full $2^n$-dimensional space is not realistic if $n$ is not very small. Fortunately, by exploiting permutation symmetry, we obtain the following  simplifications enabling an efficient calculation.

The first simplification comes from counting the number of distinct components that we need to calculate for a symmetric state. It is easy to see that $U_{\pi} \rho U_{\pi}^{\dagger} = \rho$ implies $r_{\pi(\gra)}=r_{\gra}$ for any permutation $\pi\in S_n$. Thus, it is more convenient to label each component $r_\gra$ by a 4-component vector $\vl=[\lambda_0,\lambda_1,\lambda_2,\lambda_3]$ which enumerates the numbers of subindices 0,  1, 2, and 3 in the index $\gra$. For instance, we have $r_{(0113)}=r_{(1031)}=r_{[1,2,0,1]}$. For a given $\vl$, the number of equal components $r_{\gra}$ corresponding to the same $r_{\vl}$ is given by the multinomial coefficient $\binom{n}{\lambda_0 \lambda_1 \lambda_2 \lambda_3}=\binom{n}{\vl}$.  Thus, instead of having to calculate $4^n$ components, only $\left(\binom{4}{n}\right)=\binom{4+n-1}{n}=\frac{1}{6}(n+3)(n+2)(n+1)\simeq \frac{n^3}{6}$ distinct components are needed, corresponding to the number of distinct multinomial coefficients. Note that permutation symmetry also implies that $U_{\pi} \rho U_{\pi'}^{\dagger} = \rho$ for any permutations $\pi,\pi'\in S_n$, which leads to other constraints on the components $r_\gra$.

The second simplification comes from calculating \eqref{ralphapure} in the so-called {\em Dicke basis}, which is a natural basis of the $(n+1)$-dimensional symmetric subspace and allow us to use the  states as expressed in this basis in \cite{Aulbach2010}. Any symmetric state can be decomposed as
\begin{equation*}
\ket{\psi}=\sum_{k=0}^n d_k \ket{S_k^n} 
\end{equation*}
where $\ket{S_k^n}$ are the Dicke states, $d_k\in\mathbb{C}$ and $\sum_{k=0}^n |d_k|^2=1$.
The Dicke states are written in the computational basis as
\begin{eqnarray}
\ket{S_k^n}&=&\binom{n}{k}^{-1/2}\sum_{\pi\in S_n}U_{\pi}|\underbrace{00\cdots 0}_{n-k}\underbrace{11\cdots 1}_{k}\rangle \nonumber\\
&=&\binom{n}{k}^{-1/2} \sum_{|j|=k}\ket{\grj},\label{dicke}
\end{eqnarray}
where $\grj$ is a binary vector of size $n$ such that $|\grj|\equiv\sum_{i=1}^n j_i=k$. Symmetric states can also be defined in terms of the Majorana representation \cite{Majorana1932}. In this representation, every symmetric state $\ket{\psi}$ is characterized by a collection of $n$ one-qubit states $\ket{q_i}=x_i\ket{0}+y_i\ket{1}$ which can be viewed as $n$ points in the surface of the Bloch sphere, according to the expression
\begin{equation}
\ket{\psi}=\frac{e^{i\theta}}{\mathcal{N}}\sum_{\pi\in S_n}U_{\pi}\ket{q_1}\ket{q_2}\cdots\ket{q_n},
\end{equation}
for some phase $\theta$ and some normalization factor $\mathcal{N}$. We can move from the Majorana representation to the Dicke basis thanks to the relation \cite{Bastin2009}
\begin{equation*}
d_k=\binom{n}{k}^{-1/2}\sum_{\pi\in S_n}y_{\pi(1)}\cdots y_{\pi(k)} x_{\pi(k+1)}\cdots x_{\pi(n)},
\end{equation*}
but this is really inefficient as it involves a sum over all permutations. Instead, we use the fact, also noted in \cite{Bastin2009,Martin2010}, that the Majorana parameters $z_i=x_i/y_i$ are the roots of the polynomial
\begin{equation}\label{polynome}
P(z)=\sum_{k=0}^n(-1)^k\binom{n}{k}^{-1/2}d_k\,z^k .
\end{equation}
Thus, knowing the Majorana parameters $z_i$, the Dicke components $d_k$ can be calculated in about $\mathcal{O}(n^3)$ operations by solving the set of linear equations $\sum_k A_{ik}d_k=0$, with $A$ being a matrix of entries $A_{ik}=(-1)^k\binom{n}{k}^{-1/2}z_i^k$.

The last simplification concerns the expression of the generalized Pauli matrices themselves. Indeed, in order to compute the components $r_{\vl}$ of a symmetric state of known components $d_k$ in the Dicke basis, we only need the {\em symmetric} part of the generalized Pauli matrices $\sigma_{\gra}$, that is, their projection into the symmetric subspace. These symmetric matrices are $(n+1)\times(n+1)$ matrices in the Dicke basis, which we note as $\tau_{\vl}$. Just as for the components $r_{\gra}$ of a symmetric state, there are $\binom{n}{\vl}$ generalized Pauli matrices $\sigma_{\gra}$ which are all projected onto the same symmetric matrix $\tau_{\vl}$, labeled by the index $\vl$. The matrix elements of $\tau_{\vl}$ in the Dicke basis are 
\begin{eqnarray}\label{tau}
(\tau_{\vl})_{kk'}&=&\bra{S_k^n}\sigma_{\gra}\ket{S_{k'}^n}\nonumber\\
&=&  \binom{n}{k}^{-1/2}  \binom{n}{k'}^{-1/2}  
\sum_{|\grj|=k} \, \sum_{|\grj'|=k'}\bra{\grj}\sigma_{\gra}\ket{\grj'}\nonumber\\
&=& \sum_{\substack{|\grj|=k \\|\grj'|=k}}\frac{\bra{\grj}\sigma_{0}^{\otimes \lambda_0}\otimes\sigma_{1}^{\otimes \lambda_1}\otimes\sigma_{2}^{\otimes \lambda_2}\otimes\sigma_{3}^{\otimes \lambda_3}\ket{\grj'}}{\sqrt{\binom{n}{k}\binom{n}{k'}}},\nonumber
\end{eqnarray}
where we have taken an arbitrary order for the individual Pauli matrices in $\sigma_{\gra}$. Then, the vectors $\grj$ and $\grj'$ can be cut in four pieces $\grj_\alpha$ and $\grj'_\alpha$ of size $\lambda_\alpha$ ($\alpha=[0,3]$). To simplify the notation, we note as $\mathcal{J}$ the domain satisfying the set of constraints on vectors $\grj_\alpha$ and $\grj'_\alpha$. It follows
\begin{eqnarray}
(\tau_{\vl})_{kk'}&=&\sum_{\mathcal{J}}\frac{\bra{\grj_0}\sigma_{0}^{\otimes \lambda_0}\ket{\grj'_0}\bra{\grj_1}\sigma_{1}^{\otimes \lambda_1}\ket{\grj'_1}\cdots\bra{\grj_3}\sigma_{3}^{\otimes \lambda_3}\ket{\grj'_3}}{\sqrt{\binom{n}{k}\binom{n}{k'}}}.\nonumber
\end{eqnarray}
By using the definition of the usual $2\times 2$ Pauli matrices, we can write each factor as
\begin{equation}
\left\{
\begin{array}{l}
\bra{\grj_0}\sigma_{0}^{\otimes \lambda_0}\ket{\grj'_0}=\delta_{\grj_0\,\grj'_0},\\
\bra{\grj_1}\sigma_{1}^{\otimes \lambda_1}\ket{\grj'_1}=\delta_{\grj_1\,\overline{\grj'_1}},\\
\bra{\grj_2}\sigma_{2}^{\otimes \lambda_2}\ket{\grj'_2}=(-1)^{|\grj'_2|}(i)^{\lambda_2}\delta_{\grj_2\,\overline{\grj'_2}},\\
\bra{\grj_3}\sigma_{3}^{\otimes \lambda_3}\ket{\grj'_3}=(-1)^{|\grj'_3|}\delta_{\grj_3\,\grj'_3},\\
\end{array}
\right .
\end{equation}
where {\em barred} vectors stand for the complementary vectors, for instance $\overline{(1011)}=(0100)$. These equations can be rewritten in terms of the new indices $k_\alpha=|\grj_\alpha|$ and $k'_\alpha=|\grj'_\alpha|$ as
\begin{equation}
\left\{
\begin{array}{l}
\bra{k_0}\sigma_{0}^{\otimes \lambda_0}\ket{k'_0}=\delta_{k_0\,k'_0},\\
\bra{k_1}\sigma_{1}^{\otimes \lambda_1}\ket{k'_1}=\delta_{k_1\,(\lambda_1-k'_1)},\\
\bra{k_2}\sigma_{2}^{\otimes \lambda_2}\ket{k'_2}=(-1)^{k'_2}(i)^{\lambda_2}\delta_{k_2\,(\lambda_2-k'_2)},\\
\bra{k_3}\sigma_{3}^{\otimes \lambda_3}\ket{k'_3}=(-1)^{k'_3}\delta_{k_3\,k'_3},
\end{array}
\right .
\end{equation}
and the sum over each $\grj_\alpha$ can be replaced by a sum over each $k_\alpha$ weighted by a factor $\binom{\lambda_\alpha}{k_\alpha}$. Eventually, the matrix elements $\tau_{\vl}$ can be reexpressed as
\begin{equation}\label{tau2}
(\tau_{\vl})_{kk'}=\sum\frac{i^{(2k_3+3\lambda_2 - 2k_2)}\binom{\lambda_0}{k_0}\binom{\lambda_1}{k_1}\binom{\lambda_2}{k_2}\binom{\lambda_3}{k_3}}{\sqrt{\binom{n}{k}\binom{n}{k'}}},
\end{equation}
where the sum is taken over the four indices $k_0$, $k_1$, $k_2$ and $k_3$ which can take values between 0 and $\lambda_0$, $\lambda_1$, $\lambda_2$ and $\lambda_3$, respectively, with the two constraints
\begin{equation}\label{contraintes}
\left\{
\begin{array}{lcl}
k &=& k_0 + k_1 + k_2 + k_3,\\
k'&=& k_0 + (\lambda_1 - k_1) + (\lambda_2  - k_2) + k_3.
\end{array}
\right .
\end{equation}
Note that in the worst case where each $\lambda_\alpha\simeq \floor{n/4}$, this calculation implies to calculate about $\mathcal{O}(n^4)$ terms.
As an example for two qubits ($n=2$), the symmetric part of the generalized Pauli matrices $\sigma_{(13)}$ and $\sigma_{(31)}$ corresponds to 
\begin{equation}\label{tau0101}
\tau_{[0,1,0,1]}=\left(
\begin{array}{ccc}
 0 & \frac{1}{\sqrt{2}} & 0 \\
 \frac{1}{\sqrt{2}} & 0 & -\frac{1}{\sqrt{2}} \\
 0 & -\frac{1}{\sqrt{2}} & 0
\end{array}
\right).
\end{equation}

By using explicit expressions of symmetric states with maximum geometric entanglement taken from \cite{Aulbach2010} or  3D coordinates available in the Sloane database \cite{Sloane:codetables}, these simplifications allow us to calculate efficiently the distinct components in the Dicke basis as
\begin{equation*}
r_{\vl}=\frac{1}{2^n}\bra{\psi}\sigma_{\vl}\ket{\psi}.
\end{equation*}
For $n=4$-$12$ (and also for $n=20$), the ratios between the number of zero odd components and the total number of odd components are gathered in Table \ref{table1}.
\begin{table}[h]
\begin{tabular}{|c|l|c|}
  \hline
  n & State & Zero odd/Total odd\\
  \hline
  4 & $\ket{\psi_4}=\frac{1}{\sqrt{3}}\ket{S_0^4}+\sqrt{\frac{2}{3}}\ket{S_3^4}$ & 18/25=72\% \\
  5 & $\ket{\psi_5}\simeq 0.547\ket{S_0^5}+0.837\ket{S_4^5}$ & 36/46$\simeq$ 78 \% \\
  6 & $\ket{\psi_6}=(\ket{S_1^6}+\ket{S_5^6})\sqrt{2}$ & 64/64=100 \% \\
  7 & $\ket{\psi_7}=(\ket{S_1^7}+\ket{S_6^7})\sqrt{2}$ & 90/100=90 \% \\
  8 & $\ket{\psi_8}\simeq 0.672\ket{S_1^8}+0.741\ket{S_6^8}$ & 94/130$\simeq$ 72\% \\
  9 & $\ket{\psi_9}=(\ket{S_2^9}+\ket{S_7^9})\sqrt{2}$ & 164/185$\simeq$ 89\%\\
  10 & $\ket{\psi_{10}}=(\ket{S_2^{10}}+\ket{S_8^{10}})\sqrt{2}$ & 230/230=100\% \\
  12 & $\ket{\psi_{12}}=(\ket{S_2^{10}}+\ket{S_8^{10}})\sqrt{2}$ & 341/371$\simeq$ 94\% \\
  20 & Dodecahedron state from \cite{Sloane:codetables} & 1266/1484 $\simeq$ 85\% \\
  \hline
\end{tabular}
\caption{\label{table1}Proportion of vanishing odd components in symmetric states that maximize geometric entanglement (from \cite{Aulbach2010}). Note that $\ket{\psi_{10}}$ is not the maximum but only a state really close to it which allows an explicit writing.}
\end{table}
We observe that a large proportion of odd components vanish for these states. In the special cases $n=6$ and $n=10$, really all odd components vanish. For the case $n=4$ and $n=12$, the states constructed thanks to the 3D coordinates available in \cite{Sloane:codetables} give better results, namely the ratio $24/25\simeq 96\%$ for $n=4$ and $371/371=100\%$ for $n=12$, even though these states are equivalent to those of \cite{Aulbach2010} (\ie related by symmetric unitary transformation $U^{\otimes n}$).
This is the case because the proportion of vanishing odd components is {\em basis-dependent}, in the sense that two equivalent symmetric states will have in general a different structure in the vector $r_{\vl}$ even if they have the same entanglement content.

Finally, we observe that the structure of the vector $r_{\vl}$ often takes a particularly simple form, especially for these states with a large proportion of vanishing odd components. For $n=4$ and $n=6$, all the non-zero components are given by
\begin{eqnarray}
n=4\Longrightarrow
\begin{pmatrix}
r_{\pi[4,0,0,0]}\\
r_{\pi[2,2,0,0]}\\
r_{\pi[1,1,1,1]}\\
\end{pmatrix}=\frac{1}{2^n}
\begin{pmatrix}
1\\
\pm 1/3\\
1/\sqrt{3}\\
\end{pmatrix},\label{r4}\\
n=6\Longrightarrow
\begin{pmatrix}
r_{\pi[6,0,0,0]}\\
r_{\pi[4,2,0,0]}\\
r_{\pi[2,2,2,0]}\\
\end{pmatrix}=\frac{1}{2^n}
\begin{pmatrix}
1\\
\pm 1/3\\
\pm 1/3\\
\end{pmatrix}.\label{r6}
\end{eqnarray}
In these expressions, the permutation symbols $\pi$ applied on vectors $\vl$ indicate that components with the same vector index $\vl$ up to some permutation are equal (or of opposite sign). For instance, for the 4-qubit state, it means that $r_{(0011)}=r_{(0022)}=-r_{(2211)}$. Note that the presence of the minus sign can be see as a consequence of the purity constraint \eqref{pure}.
For $n=10$ and $n=12$, the structure is similar even if a little bit more complex because it involves one or two different values per permutation of the $\vl$ index. In particular, for $n=12$, the structure is similar to that of $n=4$ or 6 in the sense that we have the non-zero components:
\begin{eqnarray}n=12\Longrightarrow
\begin{pmatrix}
r_{\pi[12,0,0,0]}\\
r_{\pi[10,2,0,0]}\\
r_{\pi[8,4,0,0]}\\
r_{\pi[8,2,2,0]}\\
\end{pmatrix}=\frac{1}{2^n}
\begin{pmatrix}
1\\
\pm 1/3\\
2/10\\
\pm 1/15\\
\end{pmatrix}.\label{r12}
\end{eqnarray}
We believe that a possible new approach to analyzing maximum entangled states in the symmetric subspace for higher values of $n$ should be inspired by these nice structures, and take the vector $r_{\vl}$ as a starting point.


\section{Conclusion}
We have investigated the maximally mixed reduction property through the concept of $k$-MM states. By making use of the generalized Bloch representation in which the partial trace operation takes a simple form, we expressed the condition that a $k$-MM state must satisfy in terms of its generalized Bloch vector components. Considering the class of balanced $k$-MM states and a connection with quantum error-correcting codes, we found asymptotic lower and upper bounds on the reduction size $k$ for $n\to \infty$. Then, we analyzed the class of symmetric states, which led us to consider a weaker version of the maximally mixed reduction property. We showed that symmetric states cannot have maximally-mixed $k$-qubit reductions with $k>1$, which is linked to the fact that some weight-two component of their generalized Bloch vector must necessarily be non-zero. In other words, symmetric states do not obey the maximally mixed reduction property (they cannot be $k$-MM states with $k$ growing linearly in $n$). However, we showed that the constraint of admitting non-zero components only holds for even components, so odd components are not constrained by permutation symmetry. We studied the case of symmetric states which maximize geometric entanglement (up to $n=20$) as examples of states admitting many zero odd components, witnessing a high multipartite entanglement content even though they do not obey the maximally mixed reduction property for $k>1$. This led us to observe some interesting structures in the Bloch vector of states maximizing the geometric entanglement, which may open new perspectives in the analysis of multipartite entanglement. 

{\em Acknowledgments:} This work was carried out with the financial support of the F.R.S-FNRS.


\end{document}